\documentclass[twocolumn,showpacs,preprintnumbers,amsmath,amssymb, superscriptaddress]{revtex4-1}

\usepackage{graphicx}
\usepackage{dcolumn}
\usepackage{dsfont}
\usepackage{amsthm}
\usepackage{bm}
\usepackage{float}
\newtheorem{theorem}{Theorem}
\newtheorem{corollary}{Corollary}[theorem]
\newtheorem{lemma}[theorem]{Lemma}

\DeclareMathOperator{\norm}{\Big|\Big|}
\DeclareMathOperator{\Li}{\mathcal{L}}

\DeclareMathOperator{\supop}{\mathcal{B}(\mathcal{B}(\mathcal{H}))}

\begin{document}

\title{Digital quantum simulation of many-body non-Markovian dynamics}

\date{\today}

\author{R. Sweke}
\email{rsweke@gmail.com}
\affiliation{Quantum Research Group, School of Physics and Chemistry, University of KwaZulu-Natal, Durban, 4001, South Africa}
\author{M. Sanz}
\affiliation{Department of Physical Chemistry, University of the Basque Country UPV/EHU, Apartado 644, E-48080 Bilbao, Spain}
\author{I. Sinayskiy}
\affiliation{Quantum Research Group, School of Physics and Chemistry, University of KwaZulu-Natal, Durban, 4001, South Africa}
\affiliation{National Institute for Theoretical Physics (NITheP), KwaZulu-Natal, South Africa}
\author{F. Petruccione}
\affiliation{Quantum Research Group, School of Physics and Chemistry, University of KwaZulu-Natal, Durban, 4001, South Africa}
\affiliation{National Institute for Theoretical Physics (NITheP), KwaZulu-Natal, South Africa}
\author{E. Solano}
\affiliation{Department of Physical Chemistry, University of the Basque Country UPV/EHU, Apartado 644, E-48080 Bilbao, Spain}
\affiliation{IKERBASQUE, Basque Foundation for Science, Maria Diaz de Haro 3, 48013 Bilbao, Spain}

\begin{abstract}

We present an algorithmic method for the digital quantum simulation of many-body locally-indivisible non-Markovian open quantum systems. It consists of two parts: Firstly, a Suzuki-Lie-Trotter decomposition of the global system propagator into the product of subsystem propagators, which may not be quantum channels, and secondly, an algorithmic procedure for the implementation of the subsystem propagators through unitary operations and measurements on a dilated space. By providing rigorous error bounds for the relevant Suzuki-Lie-Trotter decomposition, we are able to analyse the efficiency of the method, and connect it with an appropriate measure of the local indivisibility of the system. In light of our analysis, the proposed method is expected to be experimentally achievable for a variety of interesting cases.

\end{abstract}

\maketitle

\section{Introduction}\label{intro}

All quantum systems are invariably in contact with an environment to some extent. Therefore, the development of tools for the study of such open quantum systems, undergoing non-unitary dynamics as a result of system-environment interactions, is of importance for understanding a rich variety of phenomena \cite{Petruccione2002, Daley2014}. Historically, effort has been focused on studying Markovian open quantum systems, whose dynamics is described by master equations in the Gorini, Kossakowski, Sudarshan and Lindblad (GKSL) form \cite{Petruccione2002, Daley2014, Rivas2012}. However, recently there has been an explosion of interest in open quantum systems beyond the Markovian regime, in which, since the typical assumptions made in deriving GKSL master equations are no longer valid, more complex, history-dependent, descriptions of the system dynamics are necessary  \cite{Breuer2015, Rivas2014, deVega2015}. 

In particular, the study of non-Markovian open quantum systems promises to allow us to better understand the nature of dissipation and decoherence \cite{Petruccione2002, Daley2014, Rivas2012, Breuer2015, Rivas2014, deVega2015, Aolita2015, Leggio2015, Man2015, Chiuri2012, Orieux2015}, thermalisation and equilibration \cite{Znidaric2010, Kastoryano2014}, non-equilibrium phase transitions \cite{Pizorn2013, Prosen2009} and transport phenomena in strongly correlated \cite{Prosen2011, Benenti2009, Prosen2012} and biological systems \cite{Mohseni2008, Rebentrost2009, Plenio2008, Marais2013, Marais2012, Dorner2012, Huelga2013}. Furthermore, within the Markovian context, it has been shown that dissipation and decoherence, traditional enemies of quantum information processing, can be exploited as a resource for quantum computation \cite{Verstraete2009, Sinayskiy2012, Pfeiffer2015}, the preparation of topological phases \cite{Budich2014, Diehl2011, Bardyn2013} and the preparation of entangled states \cite{Kraus2008, Diehl2008}. In this sense, it is desirable to understand the extent to which these protocols are robust against relaxation of the strict assumptions involved in this setting.

Simulations on controllable quantum devices promise to be one of the most effective tools for the study of open quantum systems. While a plethora of methods have been developed for the simulation of Markovian open quantum systems, on a wide variety of quantum devices \cite{Kliesch2011, Barthel2012, Tempel2014, DiCandia2014, Wang2013, Wang2014, Wang2011, Sweke2014, Sweke2015, Dive2015, Zanardi2015, Muller2011, Schindler2013, Barreiro2011, Weimer2011, Weimer2010}, there have only recently begun to emerge proposals for either classical  \cite{Prior2010, Cerrillo2014, Rosenbach2015, Woods2015} or quantum \cite{Pfeiffer2015, Cardenas2015, Brito2015, AlvarezRodriguez2016} simulation of non-Markovian open quantum systems. If one has knowledge of certain properties of the environment, then one of the most natural approaches is through methods of embedding non-Markovian open quantum systems in larger Markovian systems \cite{deVega2015, Breuer1999, Breuer2004, Budini2013}, which can then be simulated through any of the available methods. However, inspired by the recent success of digital quantum simulations in a variety of contexts \cite{Barends2015a, Barends2015b, Salathe2015}, largely based on ``Trotterization" of the system's dynamics \cite{Trotter1959, Suzuki1990, Suzuki1991}, one may wonder about the applicability of these digital methods to a class of many-body non-Markovian dynamics, to which they appear well suited. 

In this work, we present a method for the digital quantum simulation of many-body $k$-local, locally-indivisible non-Markovian open quantum systems, rigorously defined in Section \ref{setting}. The first part of this method, described in Section \ref{trottersection}, consists in a Suzuki-Lie-Trotter (SLT) decomposition \cite{Trotter1959, Suzuki1990, Suzuki1991, LasHeras2016} of the global system propagator into the product of local propagators, which due to the local-indivisibility of the system, may not be quantum channels. Generalising the work of Ref. \cite{Kliesch2011} to this context, we provide a rigorous error bound for such an SLT decomposition, allowing us to study the efficiency of the digital method. Importantly, as discussed in Section \ref{trottersection}, we cannot expect to obtain an efficient method for an arbitrary non-Markovian open quantum system, and as such the primary aim of this analysis is to understand how the efficiency of the method depends on appropriate measures of the indivisibility of the system of interest.

The second part of the method consists in an algorithmic procedure for the implementation of Hermiticity- and trace-preserving maps which are not quantum channels. Specifically, inspired by the notion of quantum instruments \cite{Watrous2016, Nielsen2000}, we show in Section \ref{NP} how any such map may be algorithmically implemented through unitary operations and measurements on a dilated space. In Section \ref{simsection}, we then combine the results of Sections \ref{trottersection} and \ref{NP}  to provide a complete method for the digital algorithmic simulation of $k$-local, locally-indivisible, non-Markovian open quantum systems. Furthermore, we provide in Section \ref{simsection} a detailed analysis of the efficiency of this method, which then allows for an assessment of whether a particular system may be feasibly simulated, given a specified set of resources. In light of this analysis, it is expected that the proposed method should be experimentally achievable for a variety of interesting cases, in particular those which are weakly indivisible with respect to the measures we define. Finally, in Section \ref{conc} we summarise our results and present an outlook on future directions.

\section{Setting}\label{setting}

We consider finite lattices $\Lambda$, consisting of $N$ lattice sites so that $|\Lambda| = N$. With each $x \in \Lambda$ there exists an associated finite Hilbert space $\mathcal{H}_x \simeq \mathbb{C}^{d_x}$, and we define $\mathcal{H}_X = \bigotimes_{x \in X}\mathcal{H}_x$ for all subsets $X \subset \Lambda$, and $\mathcal{H}\equiv\mathcal{H}_\Lambda$. For simplicity, we assume that $d_x = d$ for all $x \in \Lambda$. We denote the space of all bounded linear operators  $A:\mathcal{H} \rightarrow \mathcal{H}$ as $\mathcal{B}(\mathcal{H})$, and given $A \in \mathcal{B}(\mathcal{H})$, we define the support of $A$, denoted $\mathrm{supp}(A)$, as the smallest subset $X\subset \Lambda$ for which there exists a non-trivial $A_X \in \mathcal{B}(\mathcal{H}_X)$ such that $A = A_X\otimes \mathds{1}_{\Lambda/X}$. For any $X \subset \Lambda$, $\mathcal{B}_X(\mathcal{H}) \equiv \{A \in \mathcal{B}(\mathcal{H})| \mathrm{supp}(A)\subset X\}$ denotes the space of all bounded linear operators on $\mathcal{H}$ with support contained in $X$. Given a Liouvillian $\mathcal{L}: \mathcal{B}(\mathcal{H}) \rightarrow \mathcal{B}(\mathcal{H}) \in  \mathcal{B}(\mathcal{B}(\mathcal{H}))$, the support of $\mathcal{L}$ is given by $\mathrm{supp}(\mathcal{L}) \equiv \bigcup \{X \subset \Lambda | \mathcal{B}_{\Lambda/X}(\mathcal{H}) \subset \mathrm{ker}(\mathcal{L})\}$, which is the set of sites on which $\mathcal{L}$ generates a non-trivial time evolution, and $\mathds{L}_X = \{\mathcal{L} \in \mathcal{B}(\mathcal{B}(\mathcal{H})) | \mathrm{supp}(\mathcal{L}) \subset X\}$ is the set of Liouvillians with support in $X$

We are interested in $k$-local open many-body quantum systems described by time-local master equations. These are systems whose dynamics satisfies 

\begin{equation}\label{goveq}
\frac{d}{dt}\rho(t) = \mathcal{L}(t)[\rho(t)] =  \sum_{Z \subset \Lambda} \mathcal{L}_Z(t)[\rho(t)], 
\end{equation}
for some piece-wise continuous time-local Liouvillian $\mathcal{L}:\mathbb{R}^+ \rightarrow \mathcal{B}(\mathcal{B}(\mathcal{H}))$, which can be written as the sum of strictly $k$-local terms $\mathcal{L}_Z:\mathbb{R}^+ \rightarrow  \mathds{L}_Z$, with $Z \subset \Lambda$. Here, strict $k$-locality means that  $|Z| \leq k$ for all $\mathcal{L}_Z$ such that $\mathcal{L}_Z (t)\neq 0$ for all $t$ - i.e. each $\mathcal{L}_Z$ term of the Liouvillian acts non-trivially on at most $k$ subsystems. Given a system defined by Eq. \eqref{goveq}, we denote by $K = |\{\mathcal{L}_{Z}(t)|\mathcal{L}_Z(t) \neq 0\}| \leq N^k$, the number of strictly $k$-local terms in the decomposition of $\mathcal{L}$. Labelling the $K$ non-trivial strictly $k$-local Liouvillians then allows us to redefine Eq. \eqref{goveq} as 

\begin{equation}\label{goveq2}
\frac{d}{dt}\rho(t) = \mathcal{L}(t)[\rho(t)] =  \sum_{i = 1}^K \mathcal{L}_i(t)[\rho(t)]. 
\end{equation}
We then define the system propagators as the family of superoperators $\{T_{\mathcal{L}}(t,s)\}$ satisfying $\rho(t) = T_{\mathcal{L}}(t,s)\rho(s)$, for all $t \geq s \geq 0$. These propagators uniquely solve the initial value problem

\begin{equation}\label{IVP}
\frac{d}{dt}T_{\mathcal{L}}(t,s) = \mathcal{L}(t)T_{\mathcal{L}}(t,s), \qquad T_{\mathcal{L}}(s,s) = \mathds{1}.
\end{equation}
In addition, for each $i \in [1,K]$, we define the \textit{local} propagators $\{T_{\mathcal{L}_i}(t,s)\}$ as the family of superoperators which uniquely solve the initial value problem

\begin{equation}\label{IVPlocal}
\frac{d}{dt}T_{\mathcal{L}_i}(t,s) = \mathcal{L}_i (t)T_{\mathcal{L}_i}(t,s), \qquad T_{\mathcal{L}_i}(s,s) = \mathds{1}.
\end{equation}

In Ref. \cite{Kliesch2011}, the digital simulation of such systems has been considered, but in the case of Markovian many-body open quantum systems, where each strictly $k$-local Liouvillian can be written in the GKSL form \cite{Petruccione2002, Daley2014, Rivas2012}. More specifically, where

\begin{equation}\label{GKSLform}
\mathcal{L}_i(t)[\cdot] = -i[H_i(t), \cdot] + \sum_{j = 1}^{d^k} \gamma_{i,j}(t)\mathcal{D}(L_{i,j}(t))[\cdot], 
\end{equation} 
with
\begin{equation}\nonumber
\mathcal{D}(L_{i,j}(t))[\cdot] = L_{i,j}(t)\cdot L_{i,j}(t)^{\dagger} - \frac{1}{2}\{ L_{i,j}(t)^{\dagger}L_{i,j}(t),\cdot \}_+
\end{equation}
and with $\gamma_{i,j}(t) \geq 0 $, for all $i \in [1,K]$, $j \in [1,d^k]$ and $t \in \mathbb{R}^+$. In this case, the system is called \textit{locally divisible}, meaning that, for all $i \in [1,K]$, and for all $0 \leq s \leq t \in \mathbb{R}^+$, the local propagator $T_{\mathcal{L}_i}(t,s)$ is a quantum channel (completely positive trace preserving map) \cite{Rivas2012, Rivas2014}. In this work, we aim to go beyond this case and consider \textit{locally indivisible} dynamics described by time-local master equations, i.e. dynamics generated by a $k$-local Liouvillian as in Eq. \eqref{goveq}, but for which $T_{\mathcal{L}_i}(t,s)$ may not be a quantum channel for all $i \in [1,K]$ and for all $0 \leq s \leq t \in \mathbb{R}^+$. Time-local master equations of this type are capable of describing many non-Markovian systems \cite{Breuer2015, Rivas2014, deVega2015}, and the simplest example of such a process is given by a system whose dynamics is described by Eq. \eqref{GKSLform}, but with dissipation rates $\gamma_{i,j}(t)$ which are not necessarily positive for all $i,j$ and $t$ \cite{Petruccione2002, Rivas2012}. We also note that we do not attempt to address here the question of which $k$-local Liouvillians generate legitimate completely positive dynamics, as the simulation method given here is valid even in the case when the global dynamics is not completely positive.

In order to quantify errors made within the presented simulation scheme, we utilise the $(1\rightarrow 1)$-norm for super-operators, where in general the $(p\rightarrow q)$-norm of a super-operator $T \in \mathcal{B}(\mathcal{B}(\mathcal{H}))$ is defined as \cite{Watrous2016} 

\begin{equation}
||T||_{p\rightarrow q} := \sup_{||A||_p=1}||T(A)||_q.
\end{equation}
The $(p\rightarrow q)$-norm defined above is induced from the Schatten $p$-norm of an operator, defined as $||A||_p:= \big(\mathrm{tr}(|A|^p)\big)^{\frac{1}{p}}$, for all $A \in \mathcal{B}(\mathcal{H})$. Notice that the definition corresponds up to a factor of 1/2 with the trace distance, $\mathrm{dist}(\rho,\sigma):=\sup_{0\leq A\leq 1}\mathrm{tr}\big(A(\rho-\sigma)\big)$, arising from a physical motivation of operational distinguishability of quantum states \cite{Nielsen2000}, which is relevant when working in the Schr\"{o}dinger picture.

\section{Trotter decomposition of locally indivisible dynamics}\label{trottersection}

In line with conventional digital quantum simulation techniques \cite{Barends2015a, Barends2015b, Salathe2015}, our strategy for the simulation of locally indivisible dynamics will be to implement $T_{\mathcal{L}}(t,0)$ through stroboscopic implementations of small time slices of the strictly $k$-local propagators, formalised via a Suzuki-Lie-Trotter (SLT) decomposition of $T_{\mathcal{L}}(t,0)$ \cite{Trotter1959, Suzuki1990, Suzuki1991}. In order to evaluate the performance of this strategy, it is necessary to obtain error bounds on the relevant SLT decomposition. To this end, we aim to generalise the results obtained in Ref. \cite{Kliesch2011} to the case of locally indivisible systems. It is essential to note that we \textit{cannot} expect to obtain an efficient simulation method for arbitrary non-Markovian systems \cite{Kliesch2011}. This is largely due to the fact that, in many non-Markovian situations in which the system of interest is strongly coupled to the environment, the relevant ``system", in terms of contribution to the dynamics, is in fact the total combination of system plus environment. Therefore, it is unrealistic to expect efficient scaling with respect to the size of the system of interest. As an illustration, if efficient simulation of arbitrary non-Markovian dynamics were possible, then one could in principle imagine efficiently simulating an extremely complicated process or computation occurring in the environment, whose results can flow back into the system of interest due to the non-Markovian character of the environment.

In light of these considerations, the primary goal of our analysis will be to provide an error bound for a relevant SLT decomposition. This will allow us to understand how the efficiency of the SLT-based digital simulation method depends on various ``measures of local indivisibility" of the simulated system. Then, given a particular locally indivisible non-Markovian system, this would permit an experimentalist the ability to determine whether the resources required for such a simulation are practically feasible. 

To this end, given a super-operator $T \in \supop$, let us define the check function

\begin{equation}\nonumber
\mathrm{Ch}\big(T \big)= \begin{cases}
0, &\qquad \text{if } T \text{ is a quantum channel,} \\
1, &\qquad \text{otherwise.}
\end{cases}
\end{equation}
Note that, given a particular super-operator $T \in \supop$, the value of $\mathrm{Ch(T)}$ can in principle be determined through construction and analysis of the Choi-Jamiolkowski state \cite{Nielsen2000}. This procedure will be practical provided that the dimension of the Hilbert space $\mathcal{H}$ is relatively small. Now, given a $k$-local system specified by a Liouvillian $\mathcal{L}$, as in Eqs. \eqref{goveq} and \eqref{goveq2}, let us consider a fixed final time $t \geq 0$ and divide the time interval $[0,t]$ into $m$ subintervals of length $\Delta t \equiv t/m$, as required by any SLT scheme. Given these values of $t$ and $m$, let us then define $T^{j}_i \equiv T^{j}_{\mathcal{L}_i}(tj/m,t(j-1)/m)$ for $i \in [1,K]$ and $j \in [1,m]$. This leads to the following ``measures of local indivisibility", 

\begin{align}
\tilde{N}^m_i &= \sum_{j = 1}^m\mathrm{Ch}\Big( T^j_i \Big) \leq m, \\
\hat{N}^m_j &= \sum_{i = 1}^K\mathrm{Ch}\Big( T^j_{i} \Big) \leq K.
\end{align}
These quantities are defined such that $\tilde{N}^m_i$ measures the number of time intervals in the SLT scheme for which the propagator $T^j_{i}$ is not a quantum channel, while $\hat{N}^m_j$ measures the number of local propagators which are not quantum channels during some given time interval $[t j/m, t(j-1)/m]$. Note that, as a consequence of strict $k$-locality, it will generally be possible to calculate these measures on a conventional computer for realistic systems in which $k$ is small and independent of the total system size $N$.  Given these quantities, we then define

\begin{align}
\tilde{N}^m &= \max_{1 \leq i \leq K} \big[ \tilde{N}^m_i  \big] \leq m,\\
\hat{N}^m &= \max_{1 \leq j \leq m} \big[ \hat{N}^m_j  \big] \leq K.
\end{align}
Clearly, for locally divisible dynamics $\tilde{N}^m = \hat{N}^m = 0$. Note that all quantities defined so far depend implicitly on the discretisation factor $m$, and that naively it is possible to bound $\tilde{N}^m$ from above by $m$, which occurs in the worst case scenario when all local propagators are not quantum channels - i.e. when the system is locally ``totally indivisible". However, it is desirable to find a tighter upper bound than this and, to this end, we define

\begin{equation}
t^{\mathrm{ID}}_i = \lim_{m \rightarrow \infty} \tilde{N}^m_i \Delta t = \lim_{m \rightarrow \infty} \frac{\tilde{N}^m_i t}{m} \leq t,
\end{equation}
and $ t^{\mathrm{ID}} = \max_{1 \leq i \leq K} t^{\mathrm{ID}}_i$. Furthermore, let us define $C^m_i$, the number of ``disjoint indivisible intervals due to $\mathcal{L}_i$", via $ C^m_i = \sum_{j = 1}^{m-1}\mathrm{Seq}\big(T^j_i, T^{j+1}_i \big)$, where

\begin{equation}
\mathrm{Seq}\big(T^j_i, T^{j+1}_i \big) = \begin{cases}
1 &\quad \mathrm{Ch}(T^j) = 0 \land  \mathrm{Ch}(T^{j+1}) = 1\\
0 &\quad \text{otherwise}.\nonumber
\end{cases}
\end{equation}
In addition, let us define $C_i  = \lim_{m \rightarrow \infty} C^{m}_i$, and $\tilde{C} = 2 \max_{1 \leq i \leq K} C_i$, so that, at this stage, it is possible to see that

\begin{equation}
\tilde{N}^m \leq \min\Big( \frac{t^{\mathrm{ID}}}{\Delta t} + \tilde{C}, m  \Big).
\end{equation}
Finally, we do not want to specify a priori that our local equations of motion are in a specific form, so we specify the local quantity 

\begin{equation}
\beta = \sup_{0 \leq s \leq t} \Big( \max_{1 \leq i \leq K}\big(  || \mathcal{L}_i(s)||_{1\rightarrow 1} \big)  \Big),
\end{equation}
which allows us to state the following theorem:

\begin{theorem}\label{SLT}
Given a system whose dynamics is described by Eqs. \eqref{goveq} and \eqref{goveq2}, the error of a first order SLT decomposition of a time evolution up to time $t$ in $m$ steps is bounded by,

\begin{equation}\nonumber
\bigg|\bigg| T_{\mathcal{L}}(t,0) -  \prod_{j = 1}^m \prod_{i = 1}^K T^j_{i} \bigg|\bigg|_{1\rightarrow 1}  \leq E(m,K,\beta, t, \tilde{N}^m, \hat{N}^m), 
\end{equation}
where 
\begin{align}
E &= \frac{K^2\beta^2t^2}{m}e^{(3 + K(2 + \tilde{N}^m) + K\mathrm{min}[m,K\tilde{N}^m] + \hat{N}^m)\beta(t/m)} \nonumber\\
&\leq \frac{K^2\beta^2t^2}{m}e^{(3 + [3 + \tilde{C}]K + \tilde{C}K^2)\beta(t/m)}e^{(K + K^2)t^{\mathrm{ID}}\beta}.
\end{align}
\end{theorem}

The proof of Theorem \ref{SLT} can be found in Appendix \ref{Proof1}. Note that we also have the following important corollary:

\begin{corollary}\label{colSLT}
Given 

\begin{equation}
0 \leq \epsilon \leq \frac{(2K^2\beta t \mathrm{ln}(2)e^{(K + K^2)t^{\mathrm{ID}}\beta})}{(3 + [3 + \tilde{C}K + \tilde{C}K^2])},
\end{equation} 
then 

\begin{equation}
 \bigg|\bigg| T_{\mathcal{L}}(t,0) -  \prod_{j = 1}^m \prod_{i = 1}^K T^j_{i}\bigg|\bigg|_{1\rightarrow 1}  \leq \epsilon,
 \end{equation}
 provided $m \geq 2K\beta^2 t^2e^{(K + K^2)t^{\mathrm{ID}}\beta}/\epsilon$.
\end{corollary}

From Corollary \ref{colSLT} (also proven in Appendix \ref{Proof1}), it is clear that, as expected, in the case of locally indivisible dynamics, the number of strictly $k$-local propagators scales exponentially in $K$, and therefore potentially exponentially in $N$ because of the the relationship $K \leq N^k$, which is not necessarily saturated. However, note that, when the dynamics is locally divisible, we have that $t^{\mathrm{ID}} = 0$. Therefore, the number of local propagators  scales polynomially in $N$, reproducing the results of Ref. \cite{Kliesch2011}. We also note that it is possible to replace the strictly $k$-local propagators $T^j_i$, with the strictly $k$-local propagators $T^{j,\mathrm{avg}}_i = \mathrm{exp}(\Delta t \mathcal{L}^{j, \mathrm{avg}}_i)$ of the \textit{averaged} Liouvillians,

\begin{equation}
\mathcal{L}^{j, \mathrm{avg}}_i = \frac{m}{t}\int_{t(j-1)/m}^{tj/m} \mathcal{L}_i(s)ds,
\end{equation}
without changing the scaling of the SLT error \cite{Kliesch2011}. Furthermore, when the Liouvillian is in GKSL form given in Eq. \eqref{GKSLform}, but possibly with negative dissipation rates at certain time, the SLT error can be expressed in terms of

\begin{equation}
\tilde{\beta} = \sup_{0 \leq s \leq t} \Big( \max_{1 \leq i \leq K}\big( \max_{1 \leq j \leq d^k} \big( || L_{i,j}(s)||_{\infty} \big)  \Big),
\end{equation}
the largest operator norm of the Lindblad operators \cite{Kliesch2011}.

At this stage, the strategy in the locally divisible case is clear, as each strictly $k$-local propagator, which is a quantum channel, can be implemented through a unitary Stinespring dilation requiring an ancilla space whose dimension depends only on $k$ and $d$ \cite{Kliesch2011, Nielsen2000, Watrous2016}. However, in the locally indivisible case, \textit{not all} local propagators are quantum channels (or even positive maps) and therefore, any realisation of an SLT scheme, such as the one provided by Theorem \ref{SLT}, requires a method for the implementation of non-positive maps.

\section{Algorithmic implementation of non-positive maps}\label{NP}

In this section, we construct a method to implement the strictly $k$-local propagators emerging from the SLT decomposition given in Theorem \ref{SLT},  which are \textit{not} quantum channels. In particular we restrict ourselves to Hermiticity- and trace preserving (HPTP), but not necessarily positive maps. Such maps would for instance arise in the case of a $k$-local system specified by a Liouvillian in GKSL form, but with negative dissipation rates for certain time intervals. As mentioned briefly in Section \ref{trottersection}, we stress that, due to \textit{strict} $k$-locality of these propagators, the support of these maps for realistic many-body systems will be sufficiently small, so that it is possible to obtain their spectrum either analytically or numerically.

Given an HPTP map $T:\mathcal{B}(\mathcal{H})\rightarrow \mathcal{B}(\mathcal{H})$, there always exists completely positive, but not necessarily trace preserving (CPnTP), maps $T^{(0)}$ and $T^{(1)}$ such that $T = T^{(0)} - T^{(1)}$. This can be proven via the spectral decomposition of the associated Choi-Jamiolkowski state \cite{Watrous2016}. As a result, we see that, if one can implement the CPnTP maps $T^{(0)}$ and $T^{(1)}$, then one can implement $T$ algorithmically. Specifically, given any initial state $\rho \in \mathcal{B}(\mathcal{H})$ and any observable $A \in \mathcal{B}(\mathcal{H})$, and defining $\rho' = T(\rho)$, we have

\begin{align}
\langle A \rangle_{\rho'} &\equiv \mathrm{tr}[A\rho']\nonumber\\
&= \mathrm{tr}[AT^{(0)}(\rho)] - \mathrm{tr}[AT^{(1)}(\rho)] \nonumber\\
&= \mathrm{tr}[A\rho'_{(0)}] - \mathrm{tr}[A\rho'_{(1)}] \nonumber\\
&= \langle A \rangle_{\rho'_{(0)}}  - \langle A \rangle_{\rho'_{(1)}}, \nonumber
\end{align}
i.e. expectation values of the desired state $\rho'$ can be algorithmically reconstructed from the expectation values of the outputs $\rho'_{(0)}$ and $\rho'_{(1)}$ of CPnTP maps $T^{(0)}$ and $T^{(1)}$. In light of this, we are able to restrict our attention to constructing a method for the implementation of CPnTP maps.

To this end, let us consider a CPnTP map $T^{(x)}:\mathcal{B}(\mathcal{H}_S) \rightarrow \mathcal{B}(\mathcal{H}_S)$, with $\mathcal{H}_S \simeq \mathbb{C}^d$, and Kraus representation $\{K_x^{(i)}\}|_{i = 1}^{d_x}$, where $d_x \leq d^2$. Furthermore, let us define the ``gauge" $G_x$ via 

\begin{equation}
G_x = \sum_{i = 1}^{d_x} (K_x^{(i)})^{\dagger}K_x^{(i)}. 
\end{equation}
As $T^{(x)}$ is not trace preserving, we know that $G_x\neq \mathds{1}$. At this stage, we can identify two cases: case 1 is when the gauge $G_x$ is sub-normalised, $\mathds{1} - G_x \geq 0$, and case 2 is when the gauge $G_x$ is not sub-normalised. More specifically, let us define $g_x = \mathrm{max}\big[\mathrm{spec}(G_x)\big]$. From the structure of $G_x$ (Hermitian and positive semi-definite), we know that that $g_x \geq 0$. Then, we are in case 1 when $g_x \leq 1$, and in case 2 otherwise. If we are in case 2, then we can define the ``renormalised" map $\hat{T}^{(x)}$ via Kraus operators $\{\hat{K}_x^{(i)}\}|_{i = 1}^{d_x}$, where $\hat{K}_x^{(i)} = (1/\sqrt{g_x}) K_x^{(i)}$.  Let us denote the gauge of $\hat{T}^{(x)}$ as $\hat{G}_x$, and note that 

\begin{equation}
\hat{G}_x = \frac{1}{g_x}G_x,
\end{equation}
so that $\hat{G}_x$ is sub-normalised by construction. Furthermore, note that for all $\rho \in \mathcal{B}(\mathcal{H}_S)$, we have that

\begin{equation}
T^{(x)}(\rho) = g_x\hat{T}^{(x)}(\rho),
\end{equation}
so that, if we can implement $\hat{T}^{(x)}$, then $T^{(x)}$ can be implemented algorithmically.

Given this setup, the problem considered here is the following:\newline

\noindent\textbf{Problem:} Given a CPnTP map  $T^{(x)}:\mathcal{B}(\mathcal{H}_S) \rightarrow \mathcal{B}(\mathcal{H}_S)$, with $\mathcal{H}_S \simeq \mathbb{C}^d$, and an observable $A \in \mathcal{B}(\mathcal{H}_S)$, and given multiple copies of $\rho \in \mathcal{B}(\mathcal{H}_S)$ (i.e. from some standard preparation procedure or preliminary circuit), describe an algorithmic procedure which yields $\langle A\rangle_{\rho'_{(x)}}$, where
\begin{equation}
\langle A\rangle_{\rho'_{(x)}} = \mathrm{Tr}[A\rho_{(x)}'],
\end{equation}
with $\rho'_{(x)}= T^{(x)}(\rho)$.\newline

If we first restrict ourselves to case 1, then the protocol described below, inspired by the notion of quantum instruments \cite{Nielsen2000, Watrous2016}, provides a solution to the problem.\newline

\noindent\textbf{Step 1:} Construct $K_x^{(\infty)}$ such that

\begin{equation}
G_x + (K_x^{(\infty)})^{\dagger}K_x^{(\infty)} = \mathds{1} 
\end{equation}
Note that the existence of $K_x^{(\infty)}$  is guaranteed by virtue of the assumed sub-normalisation of $G_x$. Furthermore, note that through the inclusion of the additional Kraus operator $K_x^{(\infty)}$, we can extend $T^{(x)}$ to a map $T^{(x)}_e$ which is both completely positive \textit{and} trace preserving.\newline

\noindent\textbf{Step 2:} Construct the unitary operator $U_x \in \mathcal{B}(\mathbb{C}^{d_x + 1}\otimes\mathcal{H}_S )$ via

\begin{equation}\label{stine} 
U_x = \begin{pmatrix}
K_x^{(1)} & \vdots & \vdots & \vdots \\
\vdots & \vdots & \vdots& \vdots \\
K_x^{(d_x)} & \vdots &\vdots & \vdots \\
K_x^{(\infty)} & \vdots &\vdots & \vdots \\
\end{pmatrix}. 
\end{equation}
Note that $U_x$ is precisely the Stinespring dilation of $T^{(x)}_e$, the trace preserving extension of $T^{(x)}$, with dilation space $\mathcal{H}_{E_x} \simeq \mathbb{C}^{d_x + 1}$.\newline

\noindent\textbf{Step 3:} Define the set of projectors $ P^{(x)} = \{P^{(x)}_1,P^{(x)}_2\}$ via

\begin{align}
P^{(x)}_1 &= \sum_{j = 1}^{d_x}\sum_{k = 1}^{d} |j,k\rangle\langle j,k | \\
 P^{(x)}_2 &= \sum_{k = 1}^d |d_x + 1,k\rangle\langle d_x + 1,k |
\end{align}
where $\{|j,k\rangle\}|_{j = 1}^{d_x + 1}|_{k = 1}^d$ is the basis for $\mathcal{H}_{E_x}\otimes\mathcal{H}_S $ in which $U_x$ is given.\newline

\noindent\textbf{Step 4:} Note now that, if one starts with the state $|1\rangle\langle 1 | \otimes \rho$, applies the unitary $U_x$, and then performs the measurement defined by $P^{(x)}$, then the probability of obtaining ``measurement outcome 1" is given by

\begin{equation}\label{need1}
 \mathrm{Pr}_x(1) \equiv \mathrm{Tr}[P^{(x)}_1U_x(|1\rangle\langle 1 | \otimes \rho)U_x^{\dagger}P^{(x)}_1] ,
\end{equation}
in which case the reduced state of the system is

\begin{equation}\label{out1}
\rho' = \mathrm{Tr}_{E_x}\bigg[\frac{P^{(x)}_1U_x(|1\rangle\langle 1 | \otimes \rho)U_x^{\dagger}P^{(x)}_1}{\mathrm{Tr}[P^{(x)}_1U_x(|1\rangle\langle 1 | \otimes \rho)U_x^{\dagger}P^{(x)}_1]}\bigg].
\end{equation}
Furthermore, note that by construction
\begin{equation}\label{need2}
T^{(x)}(\rho) =  \mathrm{Tr}_{E_x}\big[P^{(x)}_1U_x(|1\rangle\langle 1 | \otimes \rho)U_x^{\dagger}P^{(x)}_1\big],
\end{equation}
so that we can rewrite Eq. \eqref{out1}, with the help of Eqs. \eqref{need1} and \eqref{need2}, as

\begin{equation}
\rho' = \frac{T^{(x)}(\rho)}{\mathrm{Pr}_x(1)},
\end{equation}
or alternatively
\begin{equation}\label{corrected}
\rho'_{(x)} = \mathrm{Pr}_x(1)\rho'.
\end{equation}

\noindent\textbf{Step 5 (case 1):} Finally, note that via Eq. \eqref{corrected}

\begin{align}
\langle A\rangle_{\rho'_{(x)}} &= \mathrm{Tr}[A\rho'_{(x)}]\nonumber\\
& = \mathrm{Pr}_x(1)\mathrm{Tr}[A\rho'] \nonumber\\
& =  \mathrm{Pr}_x(1)\langle A\rangle_{\rho'}.
\end{align}
Now, $\langle A\rangle_{\rho'}$ can be obtained from the state $\rho'$, which in turn can be produced through unitary evolution of a dilated system via $U_x$, followed by the measurement $P^{(x)}$, and postselecting on ``measurement outcome 1". Furthermore, the constant $\mathrm{Pr}_x(1)$  can be asymptotically obtained through repetitions of the process of unitary evolution and measurement (with the same initial state each time), by recording the proportion of ``measurement outcome 1"  to ``measurement outcome 2". To sum up, through unitary evolutions and measurements of a dilated system, it is possible to obtain algorithmically the desired value of $\langle A\rangle_{\rho'}$, provided the assumption of sub-normalisation holds.\newline

Now, let us consider the case when sub-normalisation is not satisfied, i.e. case 2. In this case, we can repeat steps 1 through 4, not for the map $T^{(x)}$, but for the ``renormalised" map $\hat{T}^{(x)}$. Finally, we slightly modify step 5, where the hats now just indicate the relevant object defined from $\hat{T}^{(x)}$, as opposed to $T^{(x)}$:\newline

\noindent\textbf{Step 5 (case 2):} Note that

\begin{align}
\langle A\rangle_{\rho'_{(x)}} &= \mathrm{Tr}[Ag_x\hat{T}^{(x)}(\rho)]\nonumber\\
& = g_x\mathrm{Pr}_x(1)\mathrm{Tr}[A\hat{\rho}'] \nonumber\\
& =  g_x\hat{\mathrm{Pr}}_x(1)\langle A\rangle_{\hat{\rho}'}.
\end{align}
Now, $\langle A\rangle_{\hat{\rho}'}$ can be obtained from the state $\hat{\rho}'$ which, again, can be produced through unitary evolution of a dilated system via $\hat{U}_x$, followed by the measurement $P^{(x)}$. Again, the constant $\hat{\mathrm{Pr}}_x(1)$  can be obtained asymptotically through repetitions of the process of unitary evolution and measurement.\newline

Clearly, for this protocol to work, it is necessary to obtain the value of the constant $\mathrm{Pr}_x(1)$. In a practical setting, it is necessary to construct some estimator $\mathrm{Pr}^{N_T}_x(1)$ for $\mathrm{Pr}_x(1)$ from a finite number of measurements $N_T$. The error in approximating the desired output state $\rho'_{(x)} = \mathrm{Pr}_x(1)\rho'$ with $\tilde{\rho}'_{(x)} = \mathrm{Pr}^{N_T}_x(1)\rho'$ is then given by 

\begin{align}
||\rho'_{(x)} - \tilde{\rho}'_{(x)} || &= |\mathrm{Pr}_x(1) -  \mathrm{Pr}^{N_T}_x(1)| || \rho'|| \\ &= |\mathrm{Pr}_x(1) -  \mathrm{Pr}^{N_T}_x(1)|.
\end{align}
Therefore, given some error threshold $ \epsilon \geq 0$, it is necessary to determine the minimum number of repetitions of the process of unitary evolution and measurement which are necessary to construct an estimator $\mathrm{Pr}^{N_T}_x(1)$ such that $|\mathrm{Pr}_x(1) -  \mathrm{Pr}^{N_T}_x(1)| \leq \epsilon$. Given that the measurement $P$ only has two possible outcomes, this is essentially the problem of constructing a binomial proportion confidence interval.

As discussed in Refs. \cite{Brown2001, Wallis2013}, in order to construct an interval with reliable properties for a potentially small number of trials, or a value of $\mathrm{Pr}_x(1)$ which is potentially close to either $0$ or $1$, it is necessary to use the Wilson score interval \cite{Wilson1927}. Formally, let us denote the number trials in which measurement outcome 1 is observed as $N_1$, and define the proportion $\hat{p} = N_1/N_T$. Furthermore, the maximum error associated with our estimator will be associated with some confidence level, given by the $z$-value of a standard normal distribution, and denoted here as $z$. The Wilson score interval then prescribes that the best estimate  $\mathrm{Pr}^{N_T}_x(1)$ is given by 

\begin{equation}\label{bestestimate}
\mathrm{Pr}^{N_T}_x(1) = \frac{\hat{p} + \frac{1}{2N_T}z^2}{1 + \frac{1}{N_T}z^2},
\end{equation}
with a confidence interval $[\mathrm{Pr}^{N_T}_x(1) - E_z, \mathrm{Pr}^{N_T}_x(1)  + E_z ]$, where

\begin{equation}\label{errorex}
E_z = \frac{z\sqrt{\frac{1}{N_T}\hat{p}(1 - \hat{p}) + \frac{1}{4N^2_T}z^2}}{1 + \frac{1}{N_T}z^2}.
\end{equation}
As an example, given a $z$-value $z = 4.42$, associated with a $99.99\%$ confidence \cite{Brown2001}, this means that we will have 

\begin{equation}
|\mathrm{Pr}_x(1) -  \mathrm{Pr}^{N_T}_x(1)| \leq E_z
\end{equation}
$99.99\%$ of the times in which such an estimator is constructed. Therefore, given a maximum error tolerance of $\epsilon \geq 0$, one can show, by noting that the right hand side of Eq. \eqref{errorex} is maximized for $\hat{p} = 1/2$, that 

\begin{equation}
|\mathrm{Pr}_x(1) -  \mathrm{Pr}^{N_T}_x(1)| \leq \epsilon,
\end{equation}
with the confidence level associated with $z$, provided that

\begin{equation}\label{trialcondition}
\frac{N^2_T}{N_T + z^2} \geq \frac{z^2}{4\epsilon^2}.
\end{equation}
Note, from Eqs. \eqref{bestestimate}, \eqref{errorex} and \eqref{trialcondition}, that in the large $N_T$ limit the best estimate is given by $\mathrm{Pr}^{N_T}_x(1) = \hat{p}$, and the condition given by Eq. \eqref{trialcondition} becomes

\begin{equation}
N_T \geq \frac{z^2}{4\epsilon^2},
\end{equation}
which is what one would expect from using the more intuitive Wald confidence interval \cite{Brown2001, Wallis2013}.

At this stage, we have therefore obtained a complete algorithmic procedure for the approximate implementation of an arbitrary HPTP super-operator. In the following section, we proceed to combine this technique with the results of Section \ref{trottersection}, in order to formulate a complete procedure for the simulation of $k$-local locally indivisible dynamics.

\section{Algorithmic digital simulation of locally indivisible dynamics}\label{simsection}

In this section, we present an algorithmic digital method for the implementation of 

\begin{equation}
\tilde{T} \equiv  \prod_{j = 1}^m \prod_{i = 1}^K T^j_{i},
\end{equation}
with $m$ fixed by Corollary \ref{colSLT}. In order to develop a concise notation, let us define a multi-index $\gamma = (j,i) \in [1,mK]$, such that $\tilde{T}$ can be rewritten as

\begin{equation}
\tilde{T} = \prod_{\gamma = 1}^{mK} T^{\gamma}.
\end{equation}
We stress that the $\gamma$ indexes the strictly $k$-local propagators in the SLT decomposition, and does not indicate an exponent. Given this notation, we will then say that $T^j_i$ is the $n$th non-CP map if $T^j_i$ is non-CP (i.e. $\mathrm{Ch}(T^j_i) = 1$) and $\sum_{\gamma = 1}^{(j,i)}\mathrm{Ch}(T^{\gamma}) = n$. Furthermore, if $T^{\gamma}$ is non-CP, but HPTP, as we are assuming all non-CP strictly $k$-local propagators are, then we denote the decomposition of $T^{\gamma}$ into the difference of CPnTP maps, as shown in Section \ref{NP}, via $T^{\gamma} =  T^{\gamma,0} - T^{\gamma,1}$. In addition, it will be useful for us to define $\beta (x,n)$ as the $n$th element of the binary representation of non-negative integer $x$, and 

\begin{figure*}[t] 
\includegraphics[scale = 0.6]{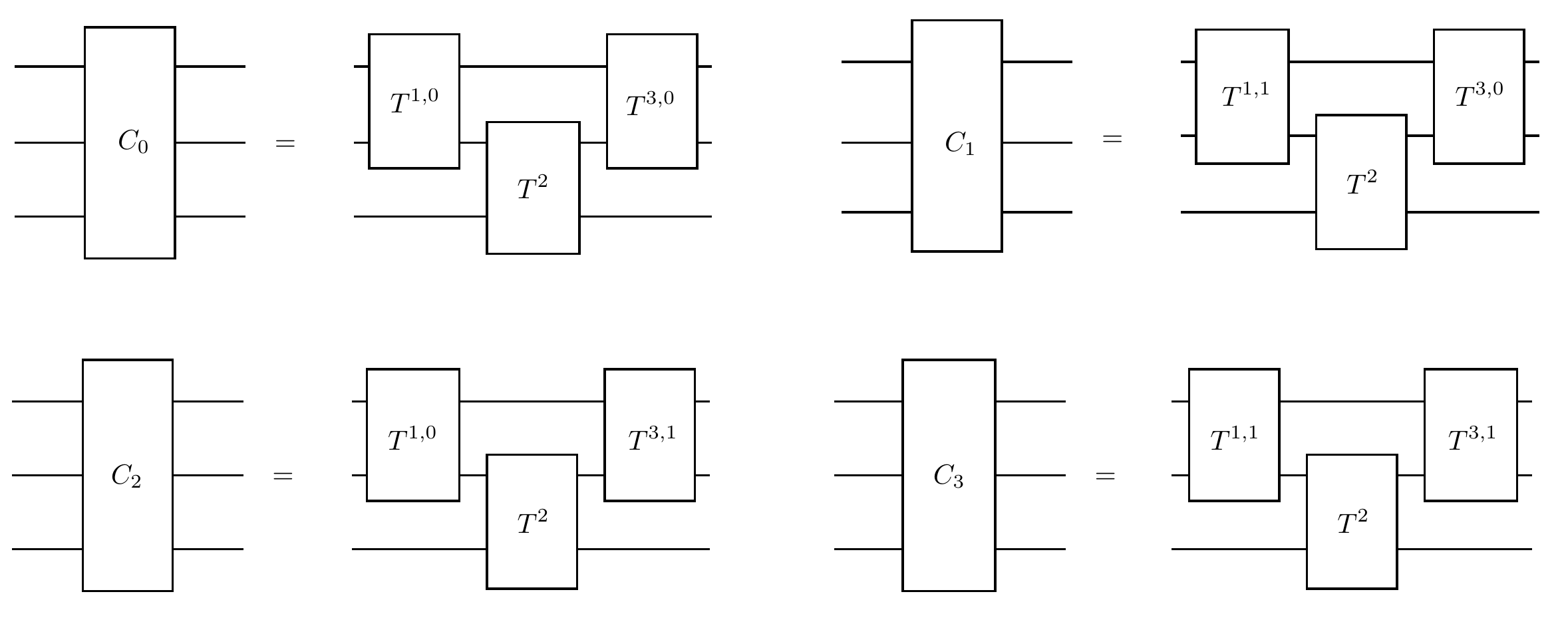}
\caption{Given a 2-local global propagator $\tilde{T} = T^3T^2T^1$, with $\mathrm{Ch}(T^1) = \mathrm{Ch}(T^3) = 1$ and $\mathrm{Ch}(T^2) = 0$, this propagator can be implemented algorithmically through the implementation of the four circuits $\{C_r\}|_{r = 0}^3$. Each circuit $C_r$ consists only of quantum channels and CPnTP maps.}\label{figure1} 
\end{figure*}

\begin{equation}
f_x(T^{\gamma}) = \begin{cases}
T^{\gamma}, & \text{if } \mathrm{Ch}(T^{\gamma}) = 0,\\
T^{\gamma, \beta(x,n)}, & \text{if } T^{\gamma} \text{ is the $n$'th non-CP map.} \nonumber
\end{cases}
\end{equation}
Defining the total number of non-CP maps appearing in the decomposition of $\tilde{T}$ as $\tilde{N}^m_{\mathrm{TOT}} = \sum_{i = 1}^K \tilde{N}^m_i$, allows us to define the $r$th circuit, denoted $C_r$ and consisting only of quantum channels and CPnTP maps, as

\begin{equation}\label{cdef}
C_r = \prod_{\gamma = 1}^{mK} f_r(T^{\gamma}),
\end{equation}
where $r \in [0,2^{\tilde{N}^m_{\mathrm{TOT}}} -1]$. Finally, defining the parity function $\mathcal{P}$ as

\begin{equation}\nonumber
\mathcal{P}(r) = 
\begin{cases}
1, & \text{if the binary representation of} \\
&\qquad \text{$r$ has an odd number of 1's,} \\
0, & \text{otherwise.} \\
\end{cases}
\end{equation}
allows us to obtain the expression

\begin{equation}\label{circuitdecomp}
\tilde{T} = \sum_{r = 0}^{2^{\tilde{N}^m_{\mathrm{TOT}}} -1} (-1)^{\mathcal{P}(r)}C_r.
\end{equation}
In essence, Eq. \eqref{circuitdecomp} shows how $\tilde{T}$ can be implemented algorithmically through the implementation of circuits consisting only of quantum channels and CPnTP maps. In other words, given an initial state $\rho(0)$ and an observable $A$, and defining $\tilde{\rho}(t) = \tilde{T}(\rho(0))$ and $\rho^{(r)}(t) = C_r(\rho(0))$, it follows from Eq. \eqref{circuitdecomp} that

\begin{align}
\langle A \rangle_{\tilde{\rho}(t)} &= \mathrm{tr}[A\tilde{\rho}(t)] \nonumber\\
& = \sum_{r = 0}^{2^{\tilde{N}^m_{\mathrm{TOT}}} -1} (-1)^{\mathcal{P}(r)}\mathrm{tr}[A\rho^{(r)}(t)] \nonumber\\
& = \sum_{r = 0}^{2^{\tilde{N}^m_{\mathrm{TOT}}} -1} (-1)^{\mathcal{P}(r)}\langle A \rangle_{\rho^{(r)}(t)}\label{alg2}, 
\end{align}
i.e. expectation values of the desired state $\tilde{\rho}(t)$ can be reconstructed from the expectation values of $\rho^{(r)}(t)$, the outputs of circuits $C_r$.

As an example, illustrated in Fig. \ref{figure1}, let us consider a two-local global propagator $\tilde{T} = T^3T^2T^1$ acting on a total system of three sites. Here, $T^2$ is a quantum channel acting non-trivially on sites 2 and 3, so $\mathrm{Ch}(T^2) = 0$, while $T^1$ and $T^3$, acting non-trivially on sites 1 and 2, are not quantum channels, so $\mathrm{Ch}(T^1) = \mathrm{Ch}(T^3) = 1$. In this simple situation. we have that the total number of non-CP maps is two. Therefore, given an initial state $\rho(0)$, the expectation values of the state $\tilde{\rho}(t) = \tilde{T}(\rho(0))$ can be reconstructed algorithmically, via Eq. \eqref{circuitdecomp}, from the states $\rho^{r}(t) = C_r(\rho(0))$ for $r \in [0:3]$. 

At this stage, what remains to be done is to incorporate explicitly into this algorithmic procedure for implementing $\tilde{T}$, the implementation of CPnTP maps within the circuits $C_r$. To this end, given a CPnTP map $T^{\gamma,i}$, with $i \in \{0,1\}$, let us denote the associated sub-normalised map as $T_s^{\gamma,i} = (1/g^{\gamma,i})T^{\gamma,i}$ and the associated CPTP extension of $T_s^{\gamma,i}$ as $T_e^{\gamma,i}$. We then denote the unitary Stinespring dilation of $T_e^{\gamma,i}$, constructed as per Eq. \eqref{stine}, as $U^{\gamma,i}$. Furthermore, given an arbitrary state $\rho$, we denote the output of a successful trial by

\begin{align}
A^{\gamma,i}(\rho) &= \frac{\mathrm{tr}_{E^{\gamma,i}} \big[P^{\gamma,i}_1U^{\gamma,i}(|1\rangle\langle 1 | \otimes \rho)(U^{\gamma,i})^{\dagger}P^{\gamma,i}_1\big]}{N^{\gamma,i}(\rho)} \nonumber\\
& = \frac{T_s^{\gamma,i}(\rho)}{N^{\gamma,i}(\rho)} \nonumber\\
& = \frac{T^{\gamma,i}(\rho)}{g^{\gamma,i}N^{\gamma,i}(\rho)}, \label{halfpoint}
\end{align}
where $P^{\gamma,i}_1$ is defined as per Step 3 of the procedure described in Section \ref{NP}, and we denote the probability of measurement outcome 1 on the input state $|1\rangle\langle 1 | \otimes \rho$ by $N^{\gamma,i}(\rho)$. Note at this stage, from Eq. \eqref{halfpoint}, that

\begin{figure*}[t] 
\includegraphics[scale = 0.7]{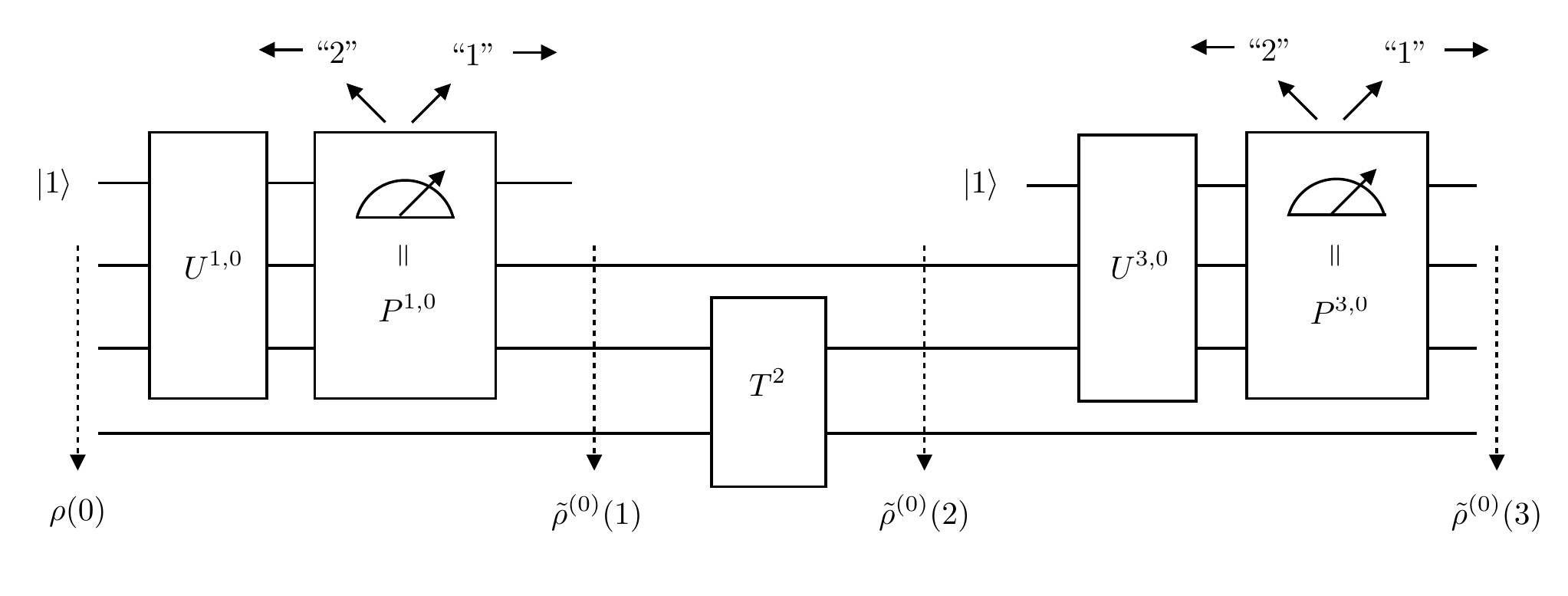}
\caption{Considering the same example as shown in Figure \ref{figure1}, the circuit shown here illustrates the method, given by Eq. \eqref{constr}, for constructing $\tilde{\rho}^{(0)}(3)$, from which the state $\rho^{(0)}(t) = C_0(\rho(0))$ can be algorithmically reconstructed. Starting with the specified initial state $\rho(0)$, the first CPnTP map $T^{1,0}$ is implemented, as described in Section \ref{NP}, through a unitary operation $U^{1,0}$ and a measurement $P^{1,0}$ on a dilated space. It is crucial to note that if measurement outcome 1 is obtained when performing the measurement, then the correct state $\tilde{\rho}^{(0)}(1)$ has been obtained and the procedure can continue, but if measurement outcome 2 is obtained then the procedure needs to be restarted. The quantum channel $T^2$ can then be implemented straightforwardly, via a conventional Stinespring dilation (not shown), before the second CPnTP map $T^{3,0}$ is implemented, analogously to $T^{1,0}$.}\label{figure2} 
\end{figure*}

\begin{equation}
T^{\gamma,i}(\rho) = g^{\gamma,i}N^{\gamma,i}(\rho)A^{\gamma,i}(\rho).
\end{equation}
Now, let us define 

\begin{equation}
M^{\gamma}_r(\rho) = \begin{cases}
T^{\gamma}(\rho), &\text{if } \mathrm{Ch}(T^{\gamma}) = 0,\\
A^{\gamma,\beta(r,n)}(\rho), &\text{if } T^{\gamma} \text{ is the $n$'th non-CP map}.\nonumber
\end{cases}
\end{equation}
Given an initial state $\rho(0)$, then 

\begin{equation}\label{form1}
\tilde{\rho}^{(r)}(j) = \prod_{\gamma = 1}^j M^{\gamma}_r(\rho(0)),
\end{equation}
such that

\begin{equation}\label{form2}
\rho^{(r)}(t) = C_r(\rho(0)) = \tilde{\rho}^{(r)}(mK)\Big( \prod_{\gamma = 1}^{mK} \mathcal{G}^{\gamma}_r \mathcal{N}^{\gamma}_r\Big),
\end{equation}
where

\begin{equation}
\mathcal{G}^{\gamma}_r = \begin{cases}
1, & \text{if } \mathrm{Ch}(T^{\gamma}) = 0,\\
1, & \text{if } T^{\gamma,\beta(r,n)} \text{is sub-normalised},\\
g^{\gamma,\beta(r,n)}, & \text{otherwise}, \nonumber
\end{cases}
\end{equation}
and

\begin{equation}
\mathcal{N}^{\gamma}_r = \begin{cases}
1, & \text{if } \mathrm{Ch}(T^{\gamma}) = 0,\\
N^{\gamma, \beta(r,n)}(\tilde{\rho}^r(\gamma -1)), & \text{otherwise}.\nonumber
\end{cases}
\end{equation}
Note that Eqs. \eqref{form1} and \eqref{form2} formalise the algorithmic procedure to implement the circuits $C_r$, by means of (a) quantum channels and (b) unitary operations and measurements on a dilated space. The quantum channels can be straightforwardly implemented via unitary Stinespring dilations \cite{Kliesch2011}. At this stage, combining Eqs. \eqref{circuitdecomp}, \eqref{form1} and \eqref{form2}, we end up with the expression

\begin{align}
\tilde{\rho}(t) &= \tilde{T}(\rho(0)) \nonumber\\
&= \Bigg(\sum_{r = 0}^{2^{\tilde{N}^m_{\mathrm{TOT}}} -1} (-1)^{\mathcal{P}(r)} \Big[ \prod_{\gamma = 1}^{mK} \mathcal{G}^{\gamma}_r \mathcal{N}^{\gamma}_r\Big]\Big[\prod_{\gamma = 1}^j M^{\gamma}_r \Big]\Bigg)\rho(0).\nonumber
\end{align}
Via a similar analysis to Eq. \eqref{alg2}, it is therefore clear that expectation values of the desired output state $\tilde{\rho}(t)$ can be algorithmically reconstructed from the expectation values of the states

\begin{equation}\label{constr}
\tilde{\rho}^{(r)}(mK) = \prod_{\gamma = 1}^{mK} M^{\gamma}_r(\rho(0)).
\end{equation}
These states can be obtained through unitary operations and measurements, involving ancillary spaces whose dimensions, independent of $N$, depend only on $d$ and $k$. Using the same example illustrated in Fig. \ref{figure1}, this procedure of obtaining $\tilde{\rho}^{(r)}(mK)$ is shown in Figure \ref{figure2}, for the case of $r = 0$. 

From Eq. \eqref{form2}, it is clear that the algorithmic reconstruction of the states $\rho^{r}(t)$ from the states $\tilde{\rho}^{(r)}(mK)$,  requires knowledge of the constants $\mathcal{N}^{\gamma}_r$. However, as discussed in Section \ref{NP}, we approximate in practice the states $\rho^{r}(t)$ with the states

\begin{equation}
\phi^{(r)}(t) = \tilde{\rho}^{(r)}(mK)\Big( \prod_{\gamma = 1}^{mK} \mathcal{G}^{\gamma}_r \tilde{\mathcal{N}}^{\gamma}_r\Big),
\end{equation}
where $\tilde{\mathcal{N}}^{\gamma}_r$ is an estimator for $\mathcal{N}^{\gamma}_r$, constructed from a finite number of measurements. The final output of the algorithmic procedure described here, an approximation of the desired state $\tilde{\rho}(t) = \tilde{T}(\rho(0))$, is therefore the state

\begin{equation}
\tilde{\phi}(t) = \sum_{r = 0}^{2^{\tilde{N}^m_{\mathrm{TOT}}} -1} (-1)^{\mathcal{P}(r)} \phi^{(r)}(t).
\end{equation}
One can then show that the algorithmic error made in approximating $\tilde{\rho}(t)$ with $\tilde{\phi}(t)$ is bounded by 

\begin{equation}
||\tilde{\rho}(t) -\tilde{\phi}(t) || \leq 2^{\tilde{N}^m_{\mathrm{TOT}} -1} \max_{r} || \rho^{(r)}(t) - \phi^{(r)}(t)||,
\end{equation}
and that the error made in approximating $\rho^{(r)}(t)$ with $\phi^{(r)}(t)$ is bounded by 

\begin{equation}\nonumber
||\rho^{(r)}(t) - \phi^{(r)}(t)|| \leq G_r \tilde{N}^{m}_{\mathrm{TOT}}\max_{1 \leq \gamma \leq mK}| \mathcal{N}^{\gamma}_r - \tilde{\mathcal{N}}^{\gamma}_r  |,
\end{equation}
with $G_r \equiv \prod_{\gamma = 1}^{mK}\mathcal{G}^{\gamma}_r$. As a result, if one requires that the total algorithmic error is less than $\epsilon_A \geq 0$, one then needs to ensure that  

\begin{equation}\label{ercond}
\max_r \Big( \max_\gamma \big( | \mathcal{N}^{\gamma}_r - \tilde{\mathcal{N}}^{\gamma}_r  |    \big)  \Big) \leq \frac{\epsilon_A}{G \tilde{N}^{m}_{\mathrm{TOT}} 2^{\tilde{N}^m_{\mathrm{TOT}} -1} },
\end{equation}
where $G \equiv \max_r \big[G_r\big]$. From Eqs. \eqref{trialcondition} and \eqref{ercond}, it is then straightforward to calculate the number of trials necessary to obtain a sufficiently accurate estimator $\tilde{\mathcal{N}}^{\gamma}_r$. 

At this stage, given an initial state $\rho(0)$, if the Trotterization error, given by Theorem \ref{SLT}, is less than $\epsilon_T$, i.e. if $\big|\big| T_{\mathcal{L}}(t,0) -  \tilde{T} \big|\big|_{1\rightarrow 1} \leq \epsilon_T$, and if the algorithmic error associated with implementing $\tilde{T}$ is less than $\epsilon_A$, then the \textit{total} error will be upper bounded by

\begin{equation}
|| \rho(t) -\tilde{\phi}(t) || \leq \epsilon_T + \epsilon_A.
\end{equation}
Therefore, if one requires a total error less than $\epsilon$, it suffices to choose $m$ such that $\epsilon_T \leq \epsilon/2$, via Corollary \ref{colSLT}, and the number of trials required for the construction of the estimators $\mathcal{N}^{\gamma}_r$, via Eqs. \eqref{trialcondition} and \eqref{ercond}, such that $\epsilon_A \leq \epsilon/2$.

Finally, it is necessary to make some comments regarding the efficiency of the method. As discussed earlier, it is not expected to obtain an efficient method for an arbitrary locally-indivisible system. Indeed, from the analysis above one can see that the number of strictly $k$-local CPnTP maps which need to be implemented, given by  $(2^{\tilde{N}^M_{\mathrm{TOT}}})mK$, where $m$ is given by Corollary \ref{colSLT}, depends strongly on the the indivisibility of the system as measured by $\tilde{N}^M_{\mathrm{TOT}}$ and $t^{\mathrm{ID}}$. Furthermore, as a result of the algorithmic procedure for implementing the non-CP strictly $k$-local propagators, each circuit $C_r$ in fact needs to be successfully implemented a number of times, given by  Eqs. \eqref{trialcondition} and \eqref{ercond}, to construct the required estimators $\tilde{\mathcal{N}}^{\gamma}_r$. However, as pointed out earlier, it is crucial to note that, as a result of the definition of $A^{\gamma,i}$, a \textit{successful implementation} of the circuit $C_r$ requires that all measurements involved in the circuit result in ``measurement outcome 1". Therefore, the probability of achieving a successful implementation of circuit $C_r$ is given by $P(C_r) = \prod_{\gamma = 1}^{mK}\mathcal{N}^{\gamma}_r$, with

\begin{equation}
\Big(\max_{1 \leq \gamma \leq mk} \big[\mathcal{N}^{\gamma}_r \big]\Big)^{\tilde{N}^{m}_{\mathrm{TOT}}} \geq P(C_r) \geq \Big(\min_{1 \leq \gamma \leq mk} \big[\mathcal{N}^{\gamma}_r \big]\Big)^{\tilde{N}^{m}_{\mathrm{TOT}}}.\nonumber
\end{equation}
In practice, as $\mathcal{N}^{\gamma, i}(\rho) \equiv \mathrm{tr}(T_s^{\gamma,i})(\rho)$, the value of $\mathcal{N}^{\gamma}_r$ can be estimated by implementing the strictly $k$-local propagator $T_s^{\gamma, \mathcal{B}(r,n)}$ on a classical computer for a random selection of inputs $\rho$, and by taking the average value of output traces. This estimated value of $\mathcal{N}^{\gamma}_r$, in conjunction with the value of $t^{\mathrm{ID}}$ and $\tilde{N}^{m}_{\mathrm{TOT}}$, can then be used to decide whether the algorithmic procedure given here is plausible for the system of interest.

\section{Conclusions and Outlook}\label{conc}

We have presented an algorithmic digital quantum simulation method for many-body locally-indivisible non-Markovian open quantum systems. The method consists of an SLT decomposition of the $k$-local global system propagator into the product of strictly $k$-local propagators, which may not be quantum channels. In this case, we also provide an algorithmic method for the implementation of those strictly $k$-local propagators which are not quantum channels, through unitary operations and measurements on a dilated space. The efficiency of the method, which reduces to the method of Ref. \cite{Kliesch2011} in the case of locally divisible dynamics, expectably depends on various measures of the local indivisibility of the system. For systems which are weakly indivisible, with respect to the measures defined here, this method should be achievable with current experimental setups \cite{Barends2015a,Barends2015b}.

In light of these results, various natural avenues arise for the extension of this work. The first direction consists in investigating any potential improvements that could be gained from utilising higher order SLT decompositions \cite{Suzuki1990, Suzuki1991}. However, as discussed in Refs. \cite{Barends2015a, Barends2015b, Salathe2015}, due to practical experimental constraints on gate implementation, any such analysis needs to take into account the tradeoff which arises between greater accuracy in the SLT decomposition and a larger number of required gates. The second natural direction involves investigating alternative or improved methods for the implementation of strictly $k$-local propagators which are not quantum channels. In particular, it would be of interest to construct methods for the simulation of maps which are not necessarily Hermiticity and trace preserving.

Finally, given the necessary inefficiency of digital methods for the simulation of non-Markovian systems, it would be of interest to investigate the potential of digital-analog approaches \cite{Arrazola2016, Mezzacapo2014, GarciaAlvarez2015}. In particular, it would be of interest to investigate whether efficient simulations are possible through the utilisation of non-Markovian analog building blocks, such as recently introduced quantum memristors \cite{Pfeiffer2015, Salmilehto2016}, combined with digital steps. Furthermore, one should investigate whether such efficient simulations could play any role in the emerging field of quantum machine learning \cite{Schuld2015}, where purely digital approaches may be restricted by fundamental obstacles.

 \begin{acknowledgments}
This work is based upon research supported by the South African Research Chair Initiative of the Department of Science and  Technology and National Research Foundation. RS acknowledges the financial support of the National Research Foundation SARChI program. MS and ES acknowledge funding from Spanish MINECO/FEDER FIS2015-69983-P and UPV/EHU UFI 11/55.
\end{acknowledgments}

\begin{appendix}
\section{Proof of Theorem 1}\label{Proof1}

In this appendix we will provide a proof for Theorem \ref{SLT}, and the associated Corollary \ref{colSLT}, through a sequence of lemmas, following the strategy given in \cite{Kliesch2011}, but generalised to the case of locally indivisible dynamics where necessary. In what follows, for notational convenience, we will drop the subscript $1\rightarrow 1$ notation from all super-operator norms, as well as the subscript 1 for operator norms. In addition, given a $k$-local system described by Eqs. \eqref{goveq} and \eqref{goveq2}, and using the same notation as in Theorem \ref{SLT}, we will define

\begin{equation}
\xi \equiv \bigg|\bigg| T_{\mathcal{L}}(t,0) -  \prod_{j = 1}^m \prod_{i = 1}^K T^j_{i} \bigg|\bigg|.
\end{equation}

Given this notation we can then state our first lemma, which will allow us to bound the norms of both local and global propagators.

\begin{lemma}\label{productintegrals}
Given $T_{\mathcal{L}}(t,s) \in \mathcal{B}(\mathcal{B}(\mathcal{H}))$, which solves the initial value problem \eqref{IVP} for some piecewise continuous Liouvillian $\mathcal{L}:\mathbb{R}^{+} \rightarrow \mathcal{B}(\mathcal{B}(\mathcal{H}))$ and some $0 \leq s \leq t$, then $\big(T_{\mathcal{L}}(t,s)\big)^{-1}$ exists, denoted $T^{-1}_{\mathcal{L}}(t,s)$, and we have that 

\begin{align}
|| T_{\mathcal{L}}(t,s)|| &\leq \mathrm{exp}\bigg[ \int_s^t || \Li(r)|| dr   \bigg]  \\  &
\leq \mathrm{exp}\bigg[ (t -s)\bigg( \sup_{s \leq r \leq t} ||\mathcal{L}(r)||  \bigg)  \bigg], \\
\end{align}
and
\begin{align}
|| T^{-1}_{\mathcal{L}}(t,s)|| &\leq \mathrm{exp}\bigg[ \int_s^t || \Li(r)|| dr   \bigg] \\   &\leq \mathrm{exp}\bigg[ (t - s)\bigg( \sup_{s\leq r \leq t} ||\mathcal{L}(r)||  \bigg)  \bigg], 
\end{align}
Furthermore, if $T_{\mathcal{L}}(t,s)$ is a quantum channel (CPTP), then we have that $||T_{\mathcal{L}}(t,s)|| = 1$.
\end{lemma}
The proof of lemma \ref{productintegrals} can be found in \cite{Kliesch2011} using properties of product integrals given in \cite{Dollard1977}. We can now proceed to begin to construct a bound on $\xi$ via the following lemma:

\begin{lemma}\label{two}
Given a $k$-local system, described by Eqs. \eqref{goveq} and \eqref{goveq2}, we have that  
 
 \begin{equation}
\xi \leq \bigg[ \sum_{j = 0}^{m-1}\bigg( \Big[ \prod_{l = j+2}^m P^{l}_2\Big]P^{j}_1\bigg)\bigg] \bigg(\max_{1 \leq j \leq m}\Big|\Big|T^j_{\mathcal{L}} - \prod_{i = 1}^KT^j_{i}\Big|\Big|\bigg),
 \end{equation}
 where
 
 \begin{equation}
 P^{\alpha}_1 = \norm \prod_{j = 1}^{\alpha}\prod_{i = 1}^K T^{j}_{i}\norm, \qquad P^{\alpha}_2 = \norm T^{\alpha}_{\Li}\norm,
\end{equation}
 and $P^0_1 = 1$.
\end{lemma}

\begin{proof} Using the same strategy as in Ref. \cite{Kliesch2011}, but taking note that $||T^j_i|| \neq 1$ for all $i,j$, we find that
\begin{align}
 \xi &= \norm \prod_{j  = 1}^mT^{j}_{\Li} - \prod_{j = 1}^m\prod_{i = 1}^K T^{j}_{i}\norm \label{startit1}\\
 & = \norm  T^m_{\Li}\prod_{j = 1}^{m-1}T^j_{\Li} - \Big(\prod_{i = 1}^K T^m_{i}\Big)\Big( \prod_{j = 1}^{m-1}\prod_{i = 1}^K T^{j}_{i}  \Big)    \norm \nonumber \\
 & \leq \norm \prod_{j = 1}^{m-1}\prod_{i = 1}^K T^{j}_{i} \norm \norm T^m_{\Li} - \prod_{i = 1}^K T^m_{i}  \norm \nonumber \\
 & \qquad+ \norm T^m_{\Li} \norm \norm \prod_{j  = 1}^{m-1}T^{j}_{\Li} - \prod_{j = 1}^{m-1}\prod_{i = 1}^K T^{j}_{i} \norm \label{ini}\\
 &\leq \sum_{j = 0}^{m-1}\bigg[ \bigg( \Big[ \prod_{l = j+2}^m P^{l}_2\Big]P^{j}_1\bigg) \Big|\Big|T^{j+1}_{\mathcal{L}} - \prod_{i = 1}^KT^{j+1}_{i}\Big|\Big|  \bigg] \label{endit1} \\
 &\leq \bigg[ \sum_{j = 0}^{m-1}\bigg( \Big[ \prod_{l = j+2}^m P^{l}_2\Big]P^{j}_1\bigg)\bigg]\bigg(\max_{1 \leq j \leq m}\Big|\Big|T^j_{\mathcal{L}} - \prod_{i = 1}^K T^j_{i}\Big|\Big|\bigg). \nonumber
 \end{align}
 In the above, \eqref{endit1} follows from \eqref{ini} by comparing the last norm on line \eqref{ini} with the right hand side of \eqref{startit1}, and then iterating. 
 \end{proof}
 
We now note that

\begin{align}
P^{\alpha}_1 &\leq \prod_{i = 1}^K\prod_{j = 1}^{m} \norm T^{j}_{i}\norm \label{st1} \\
&\leq \prod_{i = 1}^K \Big( \big[  e^{\beta\Delta t}  \big]^{\tilde{N}^m_i}  \Big) \label{en1}\\
&\leq e^{ K\beta \tilde{N}^m\Delta t } \\
& \leq e^{ K\beta m\Delta t },
\end{align}
and
\begin{align}
\prod_{i = j+2}^m P^i_2 &\leq \prod_{i = 1}^m \norm T^i_{\Li} \norm  \label{st2}\\
&\leq \prod_{i = 1}^{\min(K\tilde{N}^m,m)} e^{\beta K \Delta t} \label{en2}\\
&\leq e^{K\beta [\min(K\tilde{N}^m,m)]\Delta t} \\
& \leq e^{ K\beta m\Delta t },
\end{align}
where \eqref{en1} follows from \eqref{st1}, and \eqref{en2} follows from \eqref{st2} via Lemma \ref{productintegrals} and the definition of $\beta$. As a result of the above observations and the statement of Lemma \ref{two} we then get the following corollary:

\begin{corollary}\label{col3}
Given a $k$-local system, described by Eqs. \eqref{goveq} and \eqref{goveq2}, we have that 
\begin{align}
\xi &\leq m \bigg(  e^{[\min(K\tilde{N}^m,m) + \tilde{N}^m]K\beta\Delta t}     \bigg)\nonumber\\&\qquad
\bigg(\max_{1 \leq j \leq m}\Big|\Big|T^j_{\mathcal{L}} - \prod_{i = 1}^K T^j_{i}\Big|\Big|\bigg). \nonumber
\end{align}
\end{corollary}

From Corollary \ref{col3}, it is clear that, to proceed, it is necessary to bound the quantity $||T^j_\mathcal{L} - \prod_{i = 1}^K T^j_i||$. Such a bound is provided by the following lemma:

\begin{lemma}\label{lemint}
Using the notation and setting of Sections \ref{setting} and \ref{trottersection} we have that

\begin{align}
||T^j_\mathcal{L} - \prod_{i = 1}^K T^j_i|| & \leq \Big(Ke^{\beta \hat{N}^m_j\Delta t}\Big) \nonumber\\ &\quad \times\Big(\max_{2 \leq \phi \leq K}||T^j_{\sum_{z = 1}^{\phi}\mathcal{L}_{z}} - T^j_{\phi}T^j_{\sum_{z = 1}^{\phi - 1}\mathcal{L}_{z}}   ||\Big).\nonumber
\end{align}
\end{lemma}
Defining $\chi \equiv ||T^j_\mathcal{L} - \prod_{i = 1}^K T^j_i||$ for notational convenience, the proof proceeds as follows:

\begin{widetext}
 
\begin{proof} Taking proper account of the presence of non-CP strictly $k$-local propagators, we find that
\begin{align}
\chi &= ||T^j_{\mathcal{L}_K + \sum_{z = 1}^{K-1}\mathcal{L}_z} - T^j_K\prod_{i = 1}^{K-1}T^j_i || \label{prli1}\\
& \leq || T^j_{\mathcal{L}_K + \sum_{z = 1}^{K-1}\mathcal{L}_z} - T^j_KT^j_{\sum_{z = 1}^{K-1}\mathcal{L}_z}  + T^j_KT^j_{\sum_{z = 1}^{K-1}\mathcal{L}_z} - T^j_K\prod_{i = 1}^{K-1}T^j_i || \\
& \leq || T^j_{\mathcal{L}_K + \sum_{z = 1}^{K-1}\mathcal{L}_z} - T^j_KT^j_{\sum_{z = 1}^{K-1}\mathcal{L}_z}   || + || T^j_K || ||  T^j_{\sum_{z = 1}^{K-1}\mathcal{L}_z} - \prod_{i = 1}^{K-1}T^j_i   ||\label{prli1a} \\
& \leq \sum_{\phi = 2}^K\bigg[ \bigg( T^j_{\sum_{z = 1}^{\phi}\mathcal{L}_{z}} - T^j_{\phi}T^j_{\sum_{z = 1}^{\phi - 1}\mathcal{L}_{z}} \bigg)  \bigg(\prod_{i = \phi + 1}^K||T^j_i|| \bigg)\bigg] \label{prli2}\\
& \leq \bigg[ \sum_{\phi = 2}^K\Big(\prod_{i = \phi + 1}^K||T^j_i|| \Big)  \bigg]  \bigg[ \max_{2 \leq \phi \leq K}||T^j_{\sum_{z = 1}^{\phi}\mathcal{L}_{z}} - T^j_{\phi}T^j_{\sum_{z = 1}^{\phi - 1}\mathcal{L}_{z}}   || \bigg] \\
& \leq K \bigg[ \prod_{i = 1}^K||T^j_i||\bigg] \bigg[ \max_{2 \leq \phi \leq K}||T^j_{\sum_{z = 1}^{\phi}\mathcal{L}_{z}} - T^j_{\phi}T^j_{\sum_{z = 1}^{\phi - 1}\mathcal{L}_{z}}   ||\bigg] \label{prli3}\\
& \leq \Big(Ke^{\beta \hat{N}^m_j\Delta t}\Big) \Big(\max_{2 \leq \phi \leq K}||T^j_{\sum_{z = 1}^{\phi}\mathcal{L}_{z}} - T^j_{\phi}T^j_{\sum_{z = 1}^{\phi - 1}\mathcal{L}_{z}}   ||\Big) \label{prli4}
\end{align}

Note that line \eqref{prli2} follows from line \eqref{prli1a} by comparing the last norm on line \eqref{prli1a} with the norm on line \eqref{prli1} and iterating. Similarly, line \eqref{prli4} follows from line \eqref{prli3} via Lemma \ref{productintegrals}, the definition of $\hat{N}^m_j$ and the definition of $\beta$.
\end{proof}
We now focus our attention on bounding the quantity $||T^j_{\sum_{z = 1}^{\phi}\mathcal{L}_{z}} - T^j_{\phi}T^j_{\sum_{z = 1}^{\phi - 1}\mathcal{L}_{z}}   ||$. To this end we use the following Lemma:

\begin{lemma}\label{FTOC}
Given two arbitrary time-dependent Liouvillians $\mathcal{K}$ and $\mathcal{L}$ the following relationship holds

\begin{equation}
|| T_{\mathcal{K} + \mathcal{L}}(t,s) - T_{\mathcal{K}}(t,s)T_{\mathcal{L}}(t,s)  || \leq \frac{1}{2}(t-s)^2\bigg[ \sup_{s \leq \mu \leq r \leq t} || \big[ \mathcal{K}(u),\mathcal{L}(r)\big] || \bigg]\bigg[e^{\big[(t-s)\big(3 \sup_{s \leq \nu \leq t}||  \mathcal{K}(\nu)|| + 2 \sup_{s \leq \nu \leq t}||  \mathcal{L}(\nu)|| \big)\big] } \bigg].
\end{equation}
\end{lemma}

\begin{proof}
For notational convenience let us define $\zeta = || T_{\mathcal{K} + \mathcal{L}}(t,s) - T_{\mathcal{K}}(t,s)T_{\mathcal{L}}(t,s)  ||$. As shown in Ref. \cite{Kliesch2011}, using the fundamental theorem of calculus allows one to obtain

\begin{equation}
\zeta = \Big|\Big| \int_{s}^tdr\int_s^r d\mu T_{\mathcal{K}}(t,s)T_{\mathcal{L}}(t,r)T^{-1}_{\mathcal{K}}(\mu,s)\big[\mathcal{L}(r),\mathcal{K}(\mu)\big] T^{-1}_{\mathcal{K}}(r,\mu)T_{\mathcal{K} + \mathcal{L}}(r,s) \Big| \Big|,
\end{equation}
from which, using Lemma \ref{productintegrals}, submultiplicativity of the $1\rightarrow 1$ norm and the triangle inequality, and again noting that not all propagators are necessarily quantum channels, it follows that,

\begin{align}
\zeta &  \leq \int_{s}^tdr\int_s^r d\mu \Bigg( \big|\big| \big[\mathcal{L}(r),\mathcal{K}(\mu)\big]  \big|\big| \mathrm{exp}\Big[\Big(\int_s^t || \mathcal{K}(\nu) ||d\nu\Big) + \Big(\int_r^t || \mathcal{L}(\nu) ||d\nu\Big) \nonumber \\ &\qquad\qquad\qquad\qquad + \Big(\int_\mu^r || \mathcal{K}(\nu) ||d\nu\Big) + \Big(\int_s^r || \mathcal{K}(\nu) + \mathcal{L}(\nu) ||d\nu\Big)\Big] \Bigg)\\
& \leq \int_{s}^t dr\int_s^r d\mu \Bigg( \big|\big| \big[\mathcal{L}(r),\mathcal{K}(\mu)\big]  \big|\big| \mathrm{exp}\Big[\int_s^t \big(3||\mathcal{K}(\nu)|| + 2||\mathcal{L}(\nu)||\big)d\nu \Big]\Bigg) \\
& \leq \Bigg( \int_{s}^t dr\int_s^r d\mu \big|\big| \big[\mathcal{L}(r),\mathcal{K}(\mu)\big]  \big|\big| \Bigg) \mathrm{exp}\Big[\int_s^t \big(3||\mathcal{K}(\nu)|| + 2||\mathcal{L}(\nu)||\big)d\nu \Big] \\
& \leq \frac{1}{2}(t-s)^2\bigg[ \sup_{s \leq \mu \leq r \leq t} || \big[ \mathcal{K}(u),\mathcal{L}(r)\big] || \bigg]\bigg[e^{\big[(t-s)\big(3 \sup_{s \leq \nu \leq t}||  \mathcal{K}(\nu)|| + 2 \sup_{s \leq \nu \leq t}||  \mathcal{L}(\nu)|| \big)\big] } \bigg].
\end{align}
\end{proof}
\end{widetext}
Applying Lemma \ref{FTOC} to the special case of the norm $||T^j_{\sum_{z = 1}^{\phi}\mathcal{L}_{z}} - T^j_{\phi}T^j_{\sum_{z = 1}^{\phi - 1}\mathcal{L}_{z}}||$, then yields the corollary,

\begin{corollary}\label{colnorm}
Using the notation and setting of Sections \ref{setting} and \ref{trottersection}, we have that, for all $2 \leq \phi \leq K$, the following inequality holds

\begin{equation}\nonumber
||T^j_{\sum_{z = 1}^{\phi}\mathcal{L}_{z}} - T^j_{\phi}T^j_{\sum_{z = 1}^{\phi - 1}\mathcal{L}_{z}}||    \leq (K\beta^2)(\Delta t)^2 e^{(3 + 2K)\beta \Delta t}.
\end{equation}
\end{corollary}

Theorem \ref{SLT} now follows straightforwardly as a consequence of Corollary \ref{col3}, Lemma \ref{lemint} and Corollary \ref{colnorm}. Finally, we provide a proof of Corollary \ref{colSLT}.

\begin{proof}[Proof (Corollary \ref{colSLT}).] Assume that 

\begin{equation}\label{ass1}
0 \leq \epsilon \leq \frac{(2K^2\beta t \mathrm{ln}(2)e^{(K + K^2)t^{\mathrm{ID}}\beta})}{(3 + [3 + \tilde{C}K + \tilde{C}K^2])},
\end{equation}
and

\begin{equation}\label{ass2}
m \geq 2K\beta^2 t^2e^{(K + K^2)t^{\mathrm{ID}}\beta}/\epsilon.
\end{equation}
It follows from Theorem \ref{SLT} and assumptions \eqref{ass1} and \eqref{ass2} that

\begin{align}
\xi &\leq \frac{K^2\beta^2t^2}{m}e^{(3 + [3 + \tilde{C}]K + \tilde{C}K^2)\beta(t/m)}e^{(K + K^2)t^{\mathrm{ID}}\beta} \\
& \leq \frac{\epsilon}{2}\mathrm{exp}\Big[\frac{\big( 3 + [3 + \tilde{C}]K + \tilde{C}K^2  \big)\epsilon}{2 K^2\beta t e^{(K + K^2)t^{\mathrm{ID}}\beta} }\Big] \\
& \leq \frac{\epsilon}{2}e^{\mathrm{ln}(2)} \\
&\leq \epsilon.
\end{align}

\end{proof}

\end{appendix}

\end{document}